\documentclass[twoside]{article}

\usepackage[accepted]{aistats2020}

\usepackage{hyperref}
\usepackage{AVTcommands}
\usepackage{microtype}

\usepackage[authoryear]{AVTcitations}
\bibliography{local}

\usepackage{xpatch}
\xpatchbibmacro{cite}{\printnames{labelname}}{\printtext[bibhyperref]{\printnames{labelname}}}{}{}
\xpatchbibmacro{textcite}{\printnames{labelname}}{\printtext[bibhyperref]{\printnames{labelname}}}{}{}

\usepackage{amsthm} 

\newtheorem{theorem}{Theorem}
\newtheorem{definition}[theorem]{Definition}
\newtheorem{proposition}[theorem]{Proposition}
\newtheorem{lemma}[theorem]{Lemma}
\newtheorem{algorithm}[theorem]{Algorithm}

\newtheorem{assumption}[theorem]{Assumption}
\newtheorem{result}[theorem]{Result}
\newtheorem{example}[theorem]{Example}

\newtheorem{diagnostic}[theorem]{Diagnostic}
\newtheorem{heuristic}[theorem]{Heuristic}

\usepackage{enumerate}

\usepackage{tikz-cd}

\renewcommand{\[}{\begin{equation}}
\renewcommand{\]}{\end{equation}}
\def\<#1\>{\begin{align}#1\end{align}}
\def\?#1\?{\begin{gather}#1\end{gather}}

\begin{document}

\twocolumn[
\aistatstitle{Asynchronous Gibbs sampling}

\aistatsauthor{Alexander Terenin \And Daniel Simpson \And David Draper}

\aistatsaddress{Imperial College London \And University of Toronto \And University of California, Santa Cruz}
]

\begin{abstract}
Gibbs sampling is a Markov Chain Monte Carlo (MCMC) method often used in Bayesian learning.
MCMC methods can be difficult to deploy on parallel and distributed systems due to their inherently sequential nature.
We study \emph{asynchronous Gibbs sampling}, which achieves parallelism by simply ignoring sequential requirements.
This method has been shown to produce good empirical results for some hierarchical models, and is popular in the topic modeling community, but was also shown to diverge for other targets.
We introduce a theoretical framework for analyzing asynchronous Gibbs sampling and other extensions of MCMC that do not possess the Markov property.
We prove that asynchronous Gibbs can be modified so that it converges under appropriate regularity conditions -- we call this the \emph{exact asynchronous Gibbs} algorithm.
We study asynchronous Gibbs on a set of examples by comparing the exact and approximate algorithms, including two where it works well, and one where it fails dramatically.
We conclude with a set of heuristics to describe settings where the algorithm can be effectively used.
\end{abstract}

\section{Introduction} \label{sec:intro}

Bayesian methods have found increased application during the last two decades in the scientific community, as well as in technology, business, public policy, and other settings.
Unfortunately, Bayesian computation has become increasingly difficult as data sets and models have grown in size and complexity.

In particular, one of the standard approaches -- Markov Chain Monte Carlo (MCMC) \cite{metropolis53, hastings70} -- often does not scale well, either with data size or model complexity.
This has been addressed in recent work by deploying MCMC on parallel and distributed systems such as GPUs and compute clusters -- here we focus on the latter setting.

To use MCMC on a compute cluster efficiently, recent work has focused both on modifying the system's architecture to better suit MCMC \cite{wei15, ho13} and on modifying MCMC to better suit the system \cite{newman09}.
The simplest approach to parallelization -- running multiple parallel chains -- requires each chain to burn in individually, which limits performance.
One way to make Gibbs sampling \cite{geman84} better suited to a compute cluster is to run it \emph{asynchronously}, by sampling the next full conditional without waiting for previous ones to finish -- this is illustrated in Figure \ref{fig:ags-sketch}.
Similar approaches have recently been proposed for distributed optimization \cite{niu11}.
Unfortunately, for Gibbs sampling this creates a stochastic process that does not possess the Markov property: empirical results regarding the behavior of such processes have largely outpaced their theoretical understanding.

Asynchronous Gibbs has found widespread use and industrial deployment, especially in the natural language processing community for the Latent Dirichlet Allocation model \cite{blei03}, where it was first proposed by \textcite{newman09}, and analyzed by \textcite{ihler12}.
However, \textcite{johnson13} exhibited an explicit counterexample demonstrating that asynchronous Gibbs sampling can diverge, and analyzed its behavior for Gaussian targets.
Other authors have analyzed the algorithm in settings where conditional independence allows the Markov property to be recovered, thereby making asynchronous and synchronous execution equivalent \cite{gonzalez11, neiswanger13, terenin19a}.

The work most related to ours, and which originally appeared concurrently to our own, is that of \textcite{desa16}.
They showed -- assuming Dobrushin's condition \cite{pedersen07} -- that the asymptotic bias and mixing time of asynchronous Gibbs can be bounded.
This is the only result we are aware of that holds for general targets with fully asynchronous execution which doesn't reduce to the synchronous case.
We discuss the relationship between this important and largely complementary viewpoint with our approach in Section \ref{sec:conclusion}.

In this work, we analyze asynchronous Gibbs by defining the \emph{exact asynchronous Gibbs} algorithm that includes a correction step, which we prove converges even if executed asynchronously.
This allows study of the approximate algorithm in use by practitioners, through comparison with its convergent counterpart.
Our framework allows for convergence analysis of MCMC algorithms in settings that do not possess the Markov property, and gives a general technique for the construction of convergent algorithms.
Our contribution and focus is primarily theoretical, but we also showcase the method on a number of examples.

\section{Asynchronous Gibbs Sampling} \label{sec:async}

\begin{figure}[t!]
\[
\nonumber
\hspace*{4.5ex}
\begin{tikzcd}[column sep=3.5ex]
\mathllap{w_1: \ } x_{10} \ar{r} \ar{dr} & x_{11} \ar{r} \ar{ddr} & x_{12} \ar{r} & x_{13} \ar{r} \ar{ddr} \ar{drr} & x_{14} \ar{r} & x_{15}
\\
\mathllap{w_2: \ } x_{20} \ar{r} & x_{21} \ar{r} \ar{ur} & x_{22} \ar{r} \ar{dr} & x_{23} \ar{r} \ar{urr} & x_{24} \ar{r} & x_{25}
\\
\mathllap{w_3: \ } x_{30} \ar{r} \ar{urr} & x_{31} \ar{r} \ar{uur}  & x_{32} \ar{r} & x_{33} \ar{r} \ar{ur} & x_{34} \ar{r} & x_{35}
\end{tikzcd}
\]
\caption{One possible sampling path of asynchronous Gibbs sampling. Here, workers $w_1, w_2, w_3$ sample values and transmit them to one another. Communication is not instantaneous -- not all samples are transmitted, and not all transmissions are received.} \label{fig:ags-sketch}
\end{figure}

Asynchronous Gibbs sampling is an algorithm that modifies Gibbs sampling to make its implementation on a compute cluster more efficient.
We present it in this section informally using an \emph{actor model} \cite{hewitt73} definition of parallelism -- for an overview of the algorithm in a simpler setting, see \textcite{terenin17b}.
Formal analysis is given in Appendix \ref{apdx:conv}.
We now introduce notation.

\1[(1)] $w_i$: a worker capable of (a) performing computations, (b) transmitting messages to other workers, and (c) receiving messages from other workers.
\2 $\pi$: a probability measure with density $f(x)$ defined on $X = \R^d$ from which we wish to sample.
\3 $\m{X}^k$: a matrix where each entry $x_{ij}$ represents component $j$ of the parameter vector $\v{x}_i$ on worker $w_i$.
This gives the total state of computation on all workers at time $k$.
We generally suppress $k$ from the notation.
\0

The algorithm proceeds as follows.

\begin{algorithm}[Asynchronous Gibbs sampling] \label{alg:ags}
For all workers, repeat the following.
\1[(i)] Select a variable $x_{ij}$ from some subset of $\v{x}_i$ at random with fixed probability and update it using the full conditional distribution $x_{ij} \given \v{x}_{i,-j}$ of Gibbs sampling.
\2 Transmit $x_{ij}$ to other workers, if possible given network limitations.
\3 Update the state $\v{x}_i$ using any transmissions received from other workers by overwriting previous parameter values in arbitrary order.
\0
\end{algorithm}

This algorithm is illustrated in Figure \ref{fig:ags-sketch}.
We assume that every worker can either update or receive each component of $\v{x}_i$, and that each worker choses to do so in a random-scan fashion.
Each worker's state is based both on values it has sampled and on values it has received from other workers -- because of network delays, these may be out of date.
Thus asynchronous Gibbs with network transmission that includes the possibility of delays is not a Markov chain.

The algorithm allows for messages to be dropped or not sent entirely, making it fault-tolerant with respect to network traffic.
Since the number of possible messages grows quadratically with the number of workers, most messages will not be sent or received.
We focus here on the case where all transmitted variables are sampled via Gibbs steps.

Can anything be proven about such a process?
To study this question, we now consider two ways in which workers process updates received from other workers.

\1[(a)] \label{item:aag} \emph{Asynchronous Gibbs}: accept all updates.

\2 \label{item:eag} \emph{Exact asynchronous Gibbs}: worker $w_i$ with current state $\v{x}_i$ accepts updated parameter $x'_j$ proposed by worker $w_s$, whose state when $x'_j$ was randomly updated was $\v{x}_s$, with Metropolis-Hastings (MH) probability
\[
\label{eqn:mh-acceptance}
\min\cbr{1, \frac{f(\v{x}'_i) f(\v{x}_i \given \v{x}_s)}{f(\v{x}_i)f(\v{x}'_i \given \v{x}_s)}}
\]
where $\v{x}'_i$ is defined to be $x'_{ij} = x'_{sj}$ for component $j$ transmitted by $w_s$ and $x'_{ij} = x_{ij}$ for all other components.
\0

Note that this update rule entails transmitting both the sender's newly updated parameter, and its full state given at the time this parameter was updated.

Here, \eqref{item:aag} is the algorithm considered by \textcite{ihler12}, \textcite{desa16}, and other authors -- since it does not always converge, we refer to it as the \emph{approximate algorithm}.
We study it by examining its relationship with \eqref{item:eag}, which we call \emph{exact} because it turns out that, under appropriate regularity conditions, it converges.

\begin{theorem}[Asynchronous convergence, informal]
Let $\m{X}^k$ be an exact asynchronous Gibbs sampler.
Assume temporarily that communication is instantaneous, and suppose in this setting that $\m{X}^k$ converges sufficiently quickly on every worker as $k\to\infty$.
Then $\m{X}^k$ converges asynchronously in the sense of \textcite{baudet78, bertsekas83, frommer00}.
\end{theorem}

\begin{proof}
Appendix \ref{apdx:conv}.
\end{proof}

The argument proceeds by first studying distributed Gibbs sampling in the simpler setting where communication is instantaneous, there are no asynchronous delays, and $\m{X}^k$ again possesses the Markov property.
Here, we assume that both the underlying Gibbs sampler on each worker and communication between workers are sufficiently regular that $\m{X}^k$ converges to stationarity at sufficient rate.
We then use the general theory of \textcite{baudet78, bertsekas83, frommer00} to show that convergence under instantaneous communication implies asynchronous convergence, defined in an appropriate sense, irrespective of most details involving the asynchronous delays.

Having established this result, we proceed to study circumstances under which exact asynchronous Gibbs can be implemented.
We then study asynchronous Gibbs sampling by studying how its trajectories differ from those of exact asynchronous Gibbs.
To do so, we implement the algorithm in \emph{Scala}, a compiled language interoperable with \emph{Java} and well suited to parallel and distributed environments.
Network communication is handled by \emph{Akka}, an \emph{actor model} framework designed for large-scale distributed applications.

\subsection{Distributed computation in exchangeable latent variable models} \label{sec:exch}

If we are interested in sampling from an exchangeable latent-variable hierarchical Bayesian model, the posterior ratio used in the MH acceptance test in exact asynchronous Gibbs simplifies to an expression involving only one data point -- this means that this ratio can be evaluated locally to each worker in a parallel environment.
In particular, this allows us to partition the data and latent variables among the workers, and update a copy of the upper-level non-latent variables locally on each worker.
To illustrate, consider the model
\< \label{hierarchical-model-1}
y_i &\given \nu_i \~ A(\nu_i)
&
\nu_i \given \theta &\~ B(\theta)
&
\theta &\propto \tau(\theta)
\>
in which $A$ and $B$ are arbitrary distributions and $y_i$ are data points.
We can define a Gibbs sampler of the form
\< \label{hierarchical-gibbs-1}
\nu_i &\given \theta, y_i \~ C(\theta, y_i)
&
\theta \given \nu_1, ..,\nu_n &\~ D(\nu_1, ..,\nu_n)
\>
where $C$ and $D$ are some distributions and $n$ is the number of data points.
Assume that we can sample from $C$ directly.
Now define an asynchronous Gibbs sampler in which all workers transmit the values of their corresponding $\nu_i$ but never transmit $\theta$.
Consider a transmitted update from $\nu_j$ to $\nu'_j$.
Let $q$ be the full conditional proposal distribution on the worker that sent $\nu'_j$, and assume that this worker transmits the parameters of that distribution along with $\nu'_j$.
Since $q$ is a full conditional distribution, it does not depend on $\nu_j$  or the previous value of $\nu'_j$ on the transmitting worker.
The MH acceptance probability takes the form
\?
\min\cbr[4]{1, \frac{f(\theta, \nu'_j \given y_j) \sbr{ \prod_{i \neq j} f(\theta, \nu_i \given y_i) } q(\nu_j)}{f(\theta, \nu_j \given y_j) \sbr{ \prod_{i \neq j} f(\theta, \nu_i \given y_i) } q(\nu'_j)}} =
\\
= \min\cbr[4]{1, \frac{f(\theta, \nu'_j \given y_j) \, q(\nu_j)}{f(\theta, \nu_j \given y_j) \, q(\nu'_j)}}
\?
where $f$ is the density of the full conditional distribution in question and $q$ is the Hastings term.
Thus we can carry out the evaluation using only one data point.

If data is stored in a distributed fashion and $y_j$ is not available on other workers, we can transmit it over network along with $\nu'_j$.
If $\nu_j$ is also not available on other workers, but the latent variables $\nu_j$ form a non-overlapping partition among the workers, then we can transmit $(\nu'_j, \nu_j, y_j, q)$, because $\nu_j$ can only be updated on other workers through communication.
This situation occurs in some problems where parameters -- such as $\theta$ in Equation (\ref{hierarchical-model-1}) -- that are located at the top of a hierarchical model may depend on $\nu_j$ only through sufficient statistics, and where storing $\nu_j$ for all $j$ on every worker is thus unnecessary.

These details illustrate the flexibility that asynchronous Gibbs sampling gives the user in handling large distributed data sets.
They also show that implementing exact asynchronous Gibbs comes with substantial additional communication costs.
If the data points $y_j$ are sufficiently large, transmitting them may be too expensive.
As a compromise, we instead recommend computing and storing the MH ratios at random with small probability, and using them as a convergence diagnostic.
Remarkably, we find for many models the MH ratio is close enough to 1 sufficiently often that the MH correction does nothing the vast majority of the time, and so the approximate algorithm yields good numerical results.
This is not always the case: a target distribution where this fails is showcased in Section \ref{sec:jacobi}.
We now illustrate the algorithm on a set of examples.

\begin{figure*}
\includegraphics{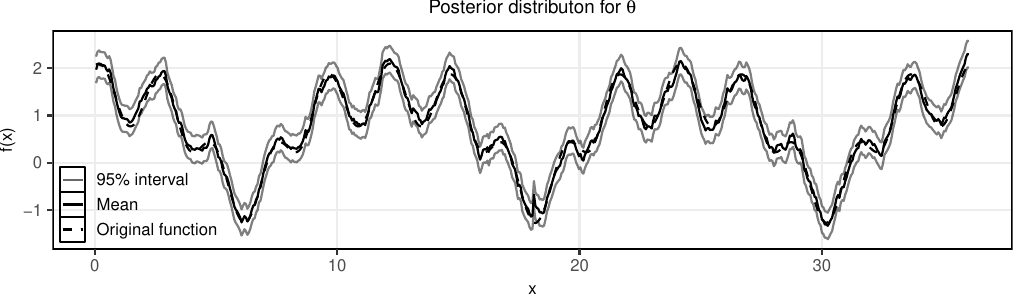}
\caption{Partial subset of $\v{\theta}$ for two workers, split at center, in the Gaussian process example of Section \ref{sec:ex2}.}
\label{gpr-results}
\end{figure*}

\section{Examples} \label{sec:examples}

\subsection{Gaussian process regression: a highly simplified spatial model with $n = 71,500$} \label{sec:ex2}

In this example, originally due to \textcite{neal97}, we used the algorithm to sample from the posterior distribution arising from a simple Gaussian process regression problem.
This example is far too simple for use in a real spatial Bayesian learning problem -- rather, we present it as a way to study how approximate asynchronous Gibbs sampling can be used for computation at scale.
Our goal was to reconstruct the function
\[
\tilde{f}(x) = 0.3 + 0.4x + 0.4\sin(2.7x) + \frac{1.1}{1 + x^2}
\]
defined for $x \in [-3, 3]$, and reflected and copied around the lines $x = 3, 9, ..$, and $x = -3, -9, ..$, in such a way that $\tilde{f}(x)$ becomes periodic with period 6 and is continuous everywhere.
To simplify our example, we assumed that our data lives on a grid with equal spacing of 0.06 (i.e., $x_1 = 0, x_2 = 0.06, x_3 = -0.06, ..$).
To generate the data, we added Gaussian white noise with standard deviation 0.2.
Our model for reconstructing this function is
\<
y_i &= f(x_i) + \eps_i\! 
& 
\!f(x_i ) &\~[GP](0,k)\!
&
\!\eps_i \~[N](0, \sigma^2)
.
\>
Here $i = 1, .., n = 71,\!500$ with $x$ on $[-2,\!145, \ 2,\!145)$.
For simplicity, we selected a Gaussian process with constant mean function $\mu$ and exponential covariance function $k(x,x') = \tau^2 \exp(- \phi| x - x'|)$ with hyperiors
\< 
\nonumber
\mu &\~[N] (a_\mu, b_\mu) 
&
\sigma^2 & \~[IG](a_\sigma, b_\sigma) 
&
\tau^2 & \~[IG](a_\tau, b_\tau) 
.
\>
By introducing latent variables $\theta_i$ corresponding to each data point, the model can be expressed as
\<
y_i & \given \theta_i, \sigma^2 \~[N] (\theta_i, \sigma^2) 
&
\v\theta &\~[N] \del{\mu \v{1}, \m{K}(\phi)}
\>
where $K_{ij}(\phi) = k(x_i,x_j)$.
By conjugacy, this yields inverse gamma posteriors for $\sigma^2$ and $\tau^2$, a Gaussian posterior for $\mu$, and a multivariate Gaussian posterior for $\v\theta$.
Since $\phi$ is non-conjugate and unidentifiable in the presence of $\tau^2$, to simplify our example we fixed it at $0.5$, which gives an interpretable length scale for the given problem.

If $n$ is large, block sampling from this posterior is intractable because it requires the frequent inversion of two $(n \x n)$ matrices.
It is possible to integrate $\v\theta$ out of the model, but this does not avoid large matrix inversion.
To avoid these difficulties and focus attention on those aspects of the computational problem most relevant to asynchronous Gibbs sampling, we use the assumption of an evenly-spaced grid together with special properties of the exponential covariance to develop a scheme for approximate closed-form analytic matrix inversion.
Details may be found in Appendix \ref{apdx:matrixinv}

With standard Gibbs, there are too many full conditionals to sample for the chain to produce useful output in reasonable time.
Asynchronous Gibbs lets us parallelize this computation.
In this example we used 143 workers with 1 CPU each.
Each worker was responsible for 500 values of $\v\theta$, each different from those handled by the other workers, and for $(\mu, \sigma^2, \tau^2)$.
We started the algorithm from low-probability initial values $\mu = 10, \sigma^2 = 10, \tau^2 = 10, \v\theta = \v{0}$.
The algorithm converged rapidly, producing approximately 10,000 samples per worker in around 20 minutes.

In Figure \ref{gpr-results} we plot a slice of the data, together with the correct solution.
As noted above, our matrix inversion approximation scheme is inaccurate around the edges of each slice of $\v\theta$ -- this can be seen in the middle of Figure \ref{gpr-results} -- and hence these values are not as accurate as those elsewhere.
The algorithm converged in an analogous fashion for all other slices of the data.
We conclude that asynchronous Gibbs sampling produces reasonable output for the given large-scale Gaussian process model.

\subsection{Mixed-effects regression: a complex hierarchical model with $n = 1,000,000$} \label{sec:ex3}

\begin{table}[b]
\begin{center}
\begin{tabular}{l r r r}
\hline
Algorithm & Workers & Threads & Runtime \\
\hline
Sequential-scan & 1 & 8 & 12 hours \\
Asynchronous & 20 & 160 & 1 hour \\
\hline
\end{tabular}
\end{center}
\caption{Runtime and degree of parallelism for 1000 Monte Carlo iterations per worker, in the hierarchical mixed-effects regression example of Section \ref{sec:ex3}.} \label{tbl:vvb-runtime}
\end{table}

The following model, due to \textcite{vonbrzeski15}, was used in a large-scale decision-theoretic analysis of product updates at eBay Inc.
Because users choose when to update to the latest version of the product, analysis of product updates is done not by controlled experiment but by observational study, and causal inference is difficult.
In particular, it is necessary to control for the \emph{early-adopter effect}, in which the behavior of the response is correlated with how quickly a user adopts the treatment after release.
To adjust for this effect, a Bayesian hierarchical mixed-effects regression model was selected.
Since we are primarily interested in the computational aspects of this problem, we omit further discussion of the particular model and evaluation of its results -- such discussion can be found in the original publication \cite{vonbrzeski15}.

A variety of different data sets have been used with this model -- the data set that we employed consists of $n = 1,000,000$ users.
The model can be written as
\?
\v{y}_i = \m{F}_i \v\beta_i + \m{W}_i \v\gamma + \v\eps_i
\\
\v\beta_i \given \v\mu, \m\Sigma \~[N](\v\mu, \m\Sigma)
\\
\v\eps_i \given \nu \~[N](\v{0}, \nu \, \m{I})
.
\?
The data set consists of $\v{y}_i: (T - p) \x 1$, $\m{F}_i: (T-p) \x d$, and $\m{W}_i: (T-p) \x (T-p)$.
The parameters are $\v\beta_i: (d \x 1 )$, $\v\gamma: (T-p) \x 1$, $\v\mu: (d \x 1)$, $\m\Sigma: (d \x d)$, and $\nu: (1\x 1)$, with the following priors:
\<
\v\mu & \~[N] (\v{0}, \kappa_\mu \m{I}_d)
&
\m\Sigma & \~[IW] (d+1, \m{I})
\\
\v\gamma & \~[N]  (\v{0}, \kappa_\gamma \m{I}_{T-p})
&
\nu & \~[IG](\epsilon / 2, \epsilon/2)
.
\>
Here $i = 1,..,n$ indexes individual data points (eBay users), $\v{y}_i$ is a vector of values representing customer satisfaction for user $i$ over time (aggregated to the weekly level), $\m{F}_i$ and $\m{W}_i$ are user-specific matrices of known constants (fixed effects), $d$ is the length of the random-effects vector, $T = 52$ is the number of weeks of data for each user, $p$ is the number of lags of autoregression in the model (typically no more than 5), and $\kappa_{\mu}, \kappa_{\gamma}, \epsilon$ are fixed hyperparameters.

The full conditionals for $\v{\mu}$, $\v{\gamma}$, $\m\Sigma$, and $\nu$ involve the full conditional sufficient statistics
\<
\v{\bar{\beta}} &= \frac{1}{n} \sum_{i=1}^n \v\beta_i
&
\m{S} &= \sum_{i=1}^n (\v\beta_i - \v{\mu}) (\v\beta_i - \v{\mu})^T
\\
\v{g} &= \sum_{i=1}^n \m{W}_i (\v{y}_i - \mathrlap{\m{F}_i \v\beta_i)}
&
&
\\
l &= \mathrlap{\sum_{i=1}^n (\v{y}_i - \m{F}_i \v\beta_i - \m{W}_i \v{\gamma})^T\!(\v{y}_i - \m{F}_i \v\beta_i - \m{W}_i \v{\gamma})}
&
&
\>
which need to be calculated in a distributed setting and broadcast to all workers.

Approximate asynchronous Gibbs can enable this computation to be performed fully in parallel by an arbitrarily large cluster, while reducing synchronization costs and improving fault tolerance.

To avoid calculations over the full data, we maintain a \emph{cache} of $\bar{\v\beta}$, $\m{S}$, $\v{g}$ and $l$.
To illustrate this, consider a new update of a single $\v\beta_i$.
When it is generated or received, the cache is updated by subtracting the portion of the sum corresponding to the old $\v\beta_i$ and adding the portion corresponding to the new value.
This significantly speeds up computation, but results in higher memory use.

Each worker updates $\v{\mu}$, $\m\Sigma$, $\v{\gamma}$, $\nu$ with the same probability as each individual element $\v\beta_i$.
With 20 workers and 1,000 iterations for each $\v\beta_i$, the algorithm generates 20(1,000) = 20,000 total samples for each variable.
This helps with mixing, improving accuracy.

For a fair performance comparison between approximate asynchronous Gibbs and standard Gibbs with multithreaded sampling of $\v\beta_{i=1, ..,n}$, we implemented a simple sequential-scan Gibbs sampler in \emph{Scala}, using the exact same multithreaded numerical routines as in our cluster sampler.
We used a data set size $n=1,000,000$, and ran for 1,000 Monte Carlo iterations per worker.
Runtime for the sequential-scan and asynchronous Gibbs samplers can be seen in Table \ref{tbl:vvb-runtime}, from which we see that asynchronous Gibbs was much faster, and scaled effectively to 20 worker nodes.

\begin{figure}
\begin{center}
\includegraphics{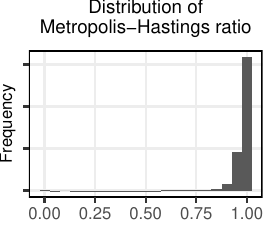}
\end{center}
\caption{Distribution of MH acceptance probability, in the hierarchical mixed-effects regression example of Section \ref{sec:ex3}.}
\label{fig:vvb-acc}
\end{figure}

Figure \ref{fig:vvb-acc} gives the distribution of the MH acceptance probabilities.
The probability of rejecting a random update is about 0.02, indicating that the behavior of the approximate algorithm is close to what the exact algorithm would have done -- this diagnostic is developed further in Section \ref{sec:jacobi}.
Both chains yielded similar diagnostic plots, which indicated issues with slow mixing.
Overall, the sequential-scan Gibbs sampler and asynchronous Gibbs appeared to have produced similar results.

It took substantially longer for $\nu$ to reach equilibrium with the asynchronous Gibbs sampler: this is a result of caching.
Before we implemented caching, the asynchronous Gibbs trace plot for $\nu$ looked similar to the sequential-scan trace plot, but the algorithm ran substantially slower due to time spent computing the relevant sum.
Trace plots for $\nu$ can be seen in Figure \ref{fig:vvb-trace} in Appendix \ref{apdx:vvb-plots}.
Note also that caching helps to ensure that all variables take a similar amount of time to sample, which is needed to ensure that the asynchronous chain is time-homogeneous -- an assumption of our analysis in Appendix \ref{apdx:conv}.

To evaluate whether the output produced by asynchronous Gibbs was meaningful, we compared it to the output produced by standard sequential-scan Gibbs sampling.
Both algorithms mix poorly, but produce similar distributional estimates for $\v\mu$ and $\v\gamma$, which were the primary unknowns of interest -- an important outcome in an unsupervised setting where cross-validation is not immediately available.
Thus the output of asynchronous Gibbs sampling was sufficient for these purposes, and in this problem the benefits of parallelism outweighed the additional implementation complexity.

\subsection{Jacobi sampling and approximate asynchronous Gibbs} \label{sec:jacobi}

\begin{figure*}
\includegraphics{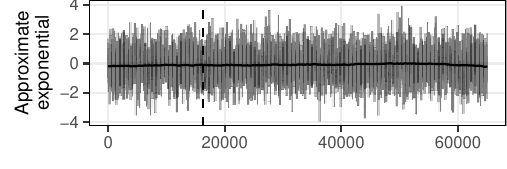}
\includegraphics{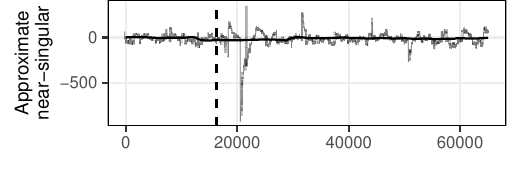}
\includegraphics{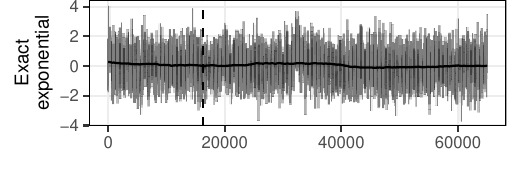}
\includegraphics{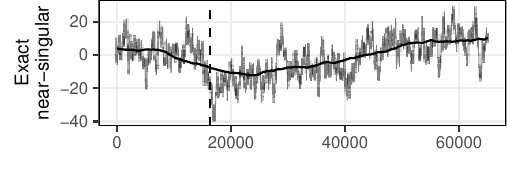}
\caption{Trace plots for the first component of $\v{x}$ for the exact and approximate variations of the asynchronous Gibbs sampler of Secton \ref{sec:jacobi} with simulated parallelism, under exponential and near-singular covariance matrices.}
\label{fig:jacobi-diag}
\end{figure*}

We now illustrate one way that asynchronous Gibbs sampling without a Metropolis-Hastings correction can fail.
This is example is originally due to \textcite{johnson13} -- we study it in the setting of exact asynchronous Gibbs, under simulated parallelism with fixed deterministic communication schemes.
Let $\pi(\v{x}) \~[N](\v{0}, \m\Sigma)$ be the target distribution.

Consider the following sampler with workers $(w_1, .., w_m)$, each of which updates one coordinate. Initialize arbitrary starting values and perform the following updates.
\1[(1)] Each $w_i$: update $x_{ii} \given \v{x}_{i,-i}$ in parallel.
\2 Each $w_i$: broadcast $x_{ii}$ to all other workers.
\0 
Here, no Metropolis-Hastings step is performed and all transmitted updates are accepted.
\textcite{johnson13} has shown that this sampling scheme does not converge for all $\m\Sigma$: it can diverge if the precision matrix $\m\Sigma^{-1}$ is not diagonally dominant.
In cases where it does converge, it's also possible for the algorithm to converge to the wrong target distribution.
We call this algorithm \emph{Jacobi sampling}, because the mean vector at each update is an iteration of the Jacobi algorithm for solving linear systems \cite{saad03} -- for the corresponding linear system, diagonal dominance suffices to ensure stability of the iterations.

We analyze the following target with $m=8$
\< \label{eqn:near-singular-matrix}
\m\Sigma^{-1} &= 1 + 0.01\m{I}
&
\m\Sigma &= \begin{cases}
\phantom{-}87.5 & i = j
\\
-12.5 & i\neq j
.
\end{cases}
\>
which we call the \emph{near-singular} covariance.
This is clearly a difficult target from the parallel sampling perspective, due to strong dependence between components.
Here, Jacobi sampling with a target that has the near-singular covariance matrix of Equation \eqref{eqn:near-singular-matrix} diverges.
For comparison, consider a correlated mean-zero 8-dimensional Gaussian, with unit exponential covariance $\Sigma_{ij} = \exp(-\phi|i - j|)$ with $\phi = 0.5$.

To study the approximate algorithm for this target, we modify the communication scheme so that the approximate algorithm numerically converges to the incorrect target, rather than diverging.
This ensures that the difference between approximate and exact algorithms is large enough to be interesting, but not so large that nothing can be said about it.

Suppose that there are 4 workers, each with 2 full conditionals assigned to them from our 8-dimensional Gaussian target.
Each worker selects one of its full conditionals at random, performs a Gibbs step, and transmits the resulting draw to each other worker with probability 0.75.
For the exact algorithm, the other workers then perform a Metropolis-Hastings calculation and either accept or reject the transmitted value.

We implemented both the exact and approximate versions of this variation with the near-singular covariance matrix on a single machine with simulated parallelism.
For comparison, we also ran the variation with an exponential covariance matrix.
Trace plots are given in Figure \ref{fig:jacobi-diag}.
Clearly the algorithm does far better with the exponential covariance.
The exact algorithm the near-singular covariance matrix mixes poorly, but ends up yielding a Monte Carlo mean and covariance matrix that are not too far away from the correct answer.
The approximate algorithm with the near-singular covariance matrix roughly yields the correct sample mean, but vastly incorrect sample covariance.

\begin{figure}[t!]
\includegraphics{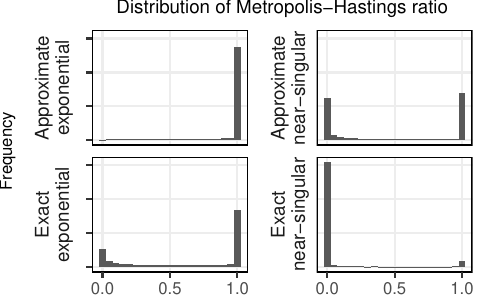}
\caption{Distribution of the Metropolis-Hastings acceptance ratio for both approximate and exact asynchronous Gibbs, with the algorithm of Section \ref{sec:jacobi}, and exponential and near-singular covariance matrices.}
\label{jacobi-acc}
\end{figure}

To further understand the differences between the exact and approximate algorithm, we examined the distribution of the MH acceptance ratios in all four examples -- these are shown in Figure \ref{jacobi-acc}.
In the case of the approximate algorithm this was accomplished by calculating and storing the MH probabilities and then ignoring them by accepting all updates.
This distribution was concentrated around 1 for the approximate exponential case.
It was substantially lower -- bimodal near 0 and 1 -- for the approximate Jacobi case that yielded the wrong answer.
Interestingly, the MH ratio distributions were also different when comparing both exact algorithms to their approximate counterparts.
This appears to be because the approximate chain undergoes \emph{phase transition} in the sense of \textcite{diaconis11}, making its behavior more akin to an optimization algorithm in those regions of the state space.

The intuition suggested by this example leads to the following diagnostic.

{
\renewcommand{\thetheorem}{A}
\begin{diagnostic} \label{diagnostic}
Approximate asynchronous Gibbs is reasonable if the distribution of the Metropolis-Hastings acceptance ratio in the approximate algorithm is concentrated around 1.
\end{diagnostic}
}

If the condition in Diagnostic \ref{diagnostic} is satisfied, the behavior of the approximate algorithm will be similar to that of the exact algorithm in the posterior regions that it explores.
Further work on approximate Markov chain theory is needed to formalize this intuition -- see Section \ref{sec:conclusion} for additional discussion.
To conclude, we provide the following heuristic for describing problems in which this is likely to occur.

{
\renewcommand{\thetheorem}{B}
\begin{heuristic} \label{heuristic}
Asynchronous Gibbs without Metropolis-Hastings correction produces a good approximation to the exact algorithm if all of the following hold:
\1[(i)] The target density $\pi$ does not possess too much dependence between its components.\label{B:1}
\2 The dimensionality of $\pi$ is significantly larger than the number of workers. \label{B:2}
\3 All transmitted variables are drawn via Gibbs steps.\label{B:3}
\0 
\end{heuristic}
}

We propose Heuristic \ref{heuristic} for the following reasons:
(\ref{B:1}) suggests that full conditional distributions in nearby posterior regions are similar, (\ref{B:2}) suggests that there is not too much movement happening at once, and (\ref{B:3}) suggests, given the previous two conditions, that the algorithm will consist of moves that are approximately Gibbs steps and hence should be accepted often.

Both Diagnostic \ref{diagnostic} and Heuristic \ref{heuristic} are intuitive tools designed to help practitioners use approximate asynchronous Gibbs in situations where it is likely to work well.
Future work is necessary to formalize these intuitions within the framework of Markov chain theory.

\section{Discussion} \label{sec:conclusion}

Asynchronous Gibbs sampling can allow Bayesian learning to be effectively implemented in parallel and distributed environments, and has been a popular choice in the topic modeling community \cite{newman09,ihler12}.
It can work well for models with a structure similar to the one found in the hierarchical mixed-effects regression example of Section \ref{sec:ex3} -- in which each data point maps to a parameter -- because the dependence in the posterior between almost all dimensions, for instance two vectors $\v\beta_i, \v\beta_j$ for $i \neq j$, is weak.
Models with strong posterior dependence will likely remain difficult for any Gibbs-based algorithm, because even in the sequential case, we expect poor mixing in that context.
One way around this would be to tailor blocking of the Gibbs sampler to the problem at hand.
For example, performance in the Gaussian process model in Section \ref{sec:ex2} could be improved by considering an overlapping block scheme, such as the \emph{additive Schwartz} method \cite{saad03}.

The theory of asynchronous Gibbs sampling can be further expanded.
While we focused on convergence, it would be useful to quantify the degree to which asynchronous delays affect the performance of the algorithm.
This is especially true for approximate asynchronous Gibbs -- we have found, surprisingly, that reducing communication latency can in some cases make performance worse.
The right amount of latency involves a balance: too much delay slows down mixing, but too little delay increases the bias introduced by ignoring the Metropolis-Hastings step.

Asynchronous Gibbs is not a Markov chain, which makes analysis non-trivial.
However, we believe that a more detailed understanding of the interplay between the convergence behavior of the asynchronous Gibbs stochastic process and its dependence on past states will be a useful step toward developing \emph{partially asynchronous} MCMC methods, which may mix better or possess other useful properties, and could potentially use asynchronous steps to hide latency during the global operations required for synchronization.
This would mirror recent advances in massively parallel iterative algorithms for solving linear systems \cite{ghysels14} and parameter-server-based distributed optimization \cite{ho13}.

Our analysis is largely complementary to the approach taken by \textcite{desa16}.
That work is based on assuming Dobrushin's condition \cite{pedersen07}, which limits their analysis to target distributions that do not exhibit too much dependence.
In contrast, our approach depends on Assumption \ref{asm:simultaneous}, which ensures all workers to converge to the target sufficiently quickly.
Both perspectives are useful: further work is needed to connect the two approaches, perhaps weakening these regularity requirements in the process.

Further work is also needed in understanding the quality of the algorithm's output, for instance by extending the standard effective sample size calculation to multiple dependent chains.
This would help verify the algorithm's output, particularly since Bayesian models are often used in unsupervised settings, such as our example in Section \ref{sec:ex3}, where algorithm-independent approaches to evaluating model quality and uncertainty, such as cross-validation, are difficult to deploy.

Implementation of asynchronous Gibbs is specific both to the problem being solved and to the hardware used -- in particular, it is necessary to decide how to divide the workload among all of the workers.
We found that different choices produced widely different mixing efficiencies -- in extreme cases, one worker can bottleneck the entire algorithm if it is sampling, at too slow a rate, a dimension upon which all other workers have strong dependence.
Similar issues can occur with respect to network traffic control: if one worker is producing output too fast, it can flood the network, preventing other workers from communicating with each other.
This is not solely an issue in complex problems -- at one point in time, due to a default \emph{Akka} configuration poorly suited to distributed computation, this difficulty manifested itself in a simple problem involving an 8-dimensional Gaussian.
Thus care was required to properly tune the algorithm in the problems we studied.

Our implementation is nowhere near optimal.
\emph{Akka} is designed for large-scale distributed web applications rather than high-performance computing.
This makes for convenient development, but does not yield the kind of low-level hardware control available in a framework such as \emph{MPI}.
Our cluster also was selected for convenience rather than performance -- indeed, the machines we used were physically located in data centers in three different US states.
This is an extremely high-latency environment from a high-performance computing perspective, and illustrates the algorithm's robustness.

These challenges are common to any nontrivial parallel computation scheme, where fully generic solutions are difficult.
Here, we have focused on studying asynchronous Gibbs under a common class of big-data Bayesian problems, for which the number of latent variables grows with the number of data points.
We find the results for hierarchical Bayesian models can mirror those of in LDA and topic modeling \cite{newman09,ihler12}, where the approach has long been popular.
Our construction with exact asynchronous Gibbs provides a view on why this algorithm has been successful in these areas.

\subsection*{Acknowledgments}

We are grateful to Chris Severs, Joseph Beynon, and Matt Gardner for many useful discussions and much help with the coding.
DS (and perhaps to a lesser extent AT and DD) would like to thank Christian Robert for sending him an early version of this paper that allowed him, during an otherwise uneventful trip to buy trousers, to remember a strategy of proof for asynchronous algorithms.
We are grateful to Matt Johnson for sending us the \emph{Jacobi sampling} example and thereby demonstrating that approximate asynchronous Gibbs is not exact, the subsequent analysis of which made much of this work possible.
We thank Christopher De Sa for sending us a different counterexample to exactness of the approximate algorithm on $\{0,1\}^3$.
We thank Murray Pollock for sending us an example demonstrating that violating time-homogeneity systematically can greatly reduce performance of the approximate algorithm.
We are grateful to Cyril Chimisov for pointing out a subtle mathematical error in our original proof.
We thank Herbie Lee for his ideas, which simplified the arguments used in our proof.
We are grateful to Gareth Roberts, Murray Pollock, Radford Neal, Daniele Venturi, Rajarshi Guhaniyogi, Tatiana Xifara, and Torsten Erhardt for interesting and useful discussions.
We also thank eBay, Inc. for providing the computational resources used for running our examples.
Membership on this list does not imply agreement with the ideas given here, nor are any of these people responsible for any errors that may be present.

\printbibliography

\newpage
\onecolumn
\appendix

\section{Appendix: convergence analysis} \label{apdx:conv}

Here we prove that, provided we start with a well-defined Markov chain, exact asynchronous Gibbs sampling will converge to the correct target distribution.
Note first that asynchronous versions of valid MCMC algorithms for 1 and 2-dimensional target distributions can be proven to always converge, because the random variables representing states of the algorithm can always be re-ordered to recover the Markov property -- see \textcite{terenin17b} for details.

Our strategy has two parts.
First, we define a serialized parallel MCMC algorithm that formalizes the way in which workers draw samples and communicate with one another under the assumption that communication is instantaneous, using ideas inspired by the coupling of chains in \emph{parallel tempering} \cite{swendsen86}.
Then, we note that MCMC methods belong to the class of \emph{fixed-point} algorithms, and hence we can use a result from the asynchronous convergence of these algorithms, due to \textcite{baudet78} and \textcite{bertsekas83}, to prove that the asynchronous version of the parallel algorithm with non-instantaneous communication converges as well.
We begin by defining the MCMC algorithm that we wish to parallelize.

\begin{definition}[Preliminaries]
Let $(\Omega,\s{F},\P)$ be a probability space.
Let $X$ be a Polish space, and let $\s{X}$ be its Borel $\sigma$-algebra.
Let $\c{M}_s(X)$ be the Banach space of signed measures on $X$, equipped with the total variation norm $\norm{\cdot}_{\f{TV}}$.
Let $\c{M}_1(X) \subset \c{M}_s(X)$ be the space of probability measures over $X$.
Let $\pi \in \c{M}_1(X)$ be the target measure.
\end{definition}

\begin{definition}[Underlying chain]
\label{def:underlying-chain}
Let $k\in\N$.
Define a Markov chain $\xi : \Omega \x \N \to X$, $(\omega,k) \mapsto \xi^k(\omega)$.
For all $B \in \s{X}$ and all $k\in\N$, let the regular conditional probability measure $P:\s{X}\x X \to \R$ defined by $P(B \given \xi) = \P(\xi^{k+1} \in B \given \xi^k = \xi)$ be its transition kernel.
This is well-defined, as the latter expression does not depend on $k$ by the Markov property and time-homogeneity.
Note that by definition of a regular conditional probability measure, for all $B\in\s{X}$ the map $\xi \mapsto P(B\given \xi)$ is $(X,\s{X})$-measurable, and for all $\xi\in X$ the map $B \mapsto P(B\given \xi)$ is a probability measure.
Define the Markov operator $P:\c{M}_1(X)\to\c{M}_1(X)$ by $(P(\mu))(B) = \int_X P(B\given \xi)\d\mu(\xi)$.
Assume that for all $\mu \in \c{M}_1(X)$, we have that $\norm{P^k(\mu) - \pi}_{\f{TV}} \to 0$ as $k \to \infty$.
We say that $\xi^k$ the \emph{underlying chain}.
\end{definition}

Here, $X$ is the parameter space for the given problem, $\mu$ is the initial measure, $\pi$ is the target measure, $k$ is the current iteration of the chain, and $P$ is the Markov operator for the chain we wish to parallelize, which we assume converges to $\pi$ in total variation.
Our analysis will center on the relationships between the workers' Markov chains, and we now introduce the definitions needed to consider this formally.

\begin{definition}[Instantaneous parallel chain] \label{def:sync_chain}
Let $m\in\N$.
Let $\c{X} = \bigtimes_{i=1}^m X$, equipped with its product $\sigma$-algebra.
Let $L$ be any index set, and let $\{\xi^k_{(l)} : l \in L\}$ be a set of underlying chains.
Let $x:\Omega\x\N\to\c{X}$, $(\omega,k) \mapsto x^k(\omega)$ be a Markov chain such that for any $x \in \c{X}$ and any $i \in \{1,..,m\}$ there exists an $l\in L$ such that for any $B \in \s{X}$ we have $\P(x^{k+1}_i \in B \given x^k = x) = \P(\xi^{k+1}_{(l)} \in B \given \xi^k_{(l)} = x_i)$.
We say that $x^k$ is the \emph{instantaneous parallel chain}.
\end{definition}

Here, $\c{X}$ is the state space for the entire compute cluster's computation.
Since we have assumed temporarily that communication is instantaneous, this means that the entire cluster's computation is also a Markov chain.
We assume that this much larger Markov chain is made up of individual components representing the worker nodes.
We also assume that each worker node performs a Markov update based on two components, namely its previous state $x_i \in X$, and a choice of proposal distribution indexed by a parameter $l\in L$, whose value can depend on the state of other workers.

\begin{example}[Instantaneous parallel Gibbs sampler]
Take $X = \R^d$ and $\c{X} = \R^{d\times m}$.
For all $i \in \{1,..,m\}$, let $C_i \subseteq \{1,..,d\}$ such that $\U_{i=1}^m C_i = \{1,..,d\}$.
Assume that $\pi$ admits an absolutely continuous density $f$ with respect to the Lebesgue measure.
Let $\m{X}$ be a Markov chain defined on $\R^{d\times m}$ as follows.
\1 Select an index $s \in \{1,..,m\}$ uniformly at random.
\2 Select a coordinate $j \in C_s$ uniformly at random.
\3 Randomly draw $x'_{sj}$ from $f(x_{sj} \given \v{x}_{s,-j})$.
\4 For all $i$, set $x_{ij}$ at the next iteration to $x'_{sj}$ with probability
\[ \label{acceptance-probability-1}
\alpha_i = \min\cbr{1,\frac{f(x'_{sj}, \v{x}_{i,-j}) \, f(x_{ij} \given \v{x}_{s,-j})}{f(x_{ij}, \v{x}_{i,-j}) \, f(x'_{sj} \given \v{x}_{s,-j})}}
\]
and set it to $x_{ij}$ otherwise.
\0 
\end{example}

This chain describes how exact asynchronous Gibbs sampling would behave under instantaneous communication, with no asynchronous delays and all messages sent and received.
It selects a worker at random and proposes from that worker's full conditional at every worker.
Note that $\alpha_{i'} = 1$, because on worker whose full conditional is selected, the proposal is exactly a Gibbs step and is hence always accepted.

It is easily seen that this is, in fact, an instantaneous parallel chain, because at every iteration it performs a Metropolis-Hastings transition with respect to some proposal distribution determined by the current state of another worker.
Thus, we can take $L$ to be the set of all such proposal distributions.

At this stage, we don't yet know anything about the stationarity properties of the instantaneous parallel chain, due to the expanded state space.
Indeed, depending on how parameters are partitioned and the details of how workers communicate, which at this stage have been abstracted out of the problem, this chain can be reducible, making stationarity analysis non-trivial.
We would therefore like to avoid speaking about the joint distribution of the chain altogether, and instead only study the marginal distributions at every worker.
To do this, we introduce a notion of coupling.

\begin{definition}[Marginally coupled Markov operator]
\label{def:m-coupled-operator}
Let $x^k$ be an instantaneous parallel chain.
Let $E = \bigtimes_{i=1}^m \c{M}_1(X)$.
For $\eps, \varpi \in E$, define the metric $d(\eps, \varpi) = \sum_{i=1}^m ||\eps_i - \varpi_i||_{\f{TV}}$.
For all $i\in\{1,..,m\}$ and all $B, B_i \in \s{X}$, define the map
\<
H_i: \s{X} \x \c{X} &\to \R
&
H_i(B \given x_1,..,x_m) = \P(x_i^{k+1} \in B \given x^k_1 = x_1,..,x^k_m=x_m)
\>
and the operator
\<
H: E &\to E
&
(H_i(\eps))(B_i) &= \int_\Omega..\int_\Omega H_i(B_i\given x_1,..,x_m) \d\eps_1(x_1) .. \d\eps_m(x_m)
.
\>
Call $H$ the \emph{marginally coupled Markov operator} for the instantaneous parallel chain.
\end{definition}

Here, $E$ is a space in which each $\eps \in E$ represents the distributional state of the entire cluster.
The marginally coupled Markov operator $H$ -- analogous to the underlying chain's Markov operator $P$ -- captures how the cluster transitions from one state to the next probabilistically, while only tracking marginal distributions rather than the full joint.
This means that we are only analyzing whether each worker converges to the target distribution, and ignoring any dependence between workers.
To continue, we need an assumption.

\begin{assumption}[Simultaneous worker-wise contraction]
\label{asm:simultaneous}
Let $\mu \in \c{M}_1(X)$. Consider the marginally coupled Markov transition kernel $H_i$.
Recall that for any fixed set of values $x^{-i}_{1:m} = \{x_j : j \in \{1,..,m\}, j \neq i\}$, by Definition \ref{def:sync_chain} there exists an $l\in L$ such that $H_i$ is the Markov transition kernel of an underlying chain $\xi_{(l)}^k$.
Let $P_{(l)}$ be the Markov operator of that chain.
Assume that for all $l$ and all $i$ there exists a $\rho < 1$ such that
\[
\norm[1]{P_{(l)}(\mu) - \pi}_{\f{TV}} \leq \rho\norm{\mu - \pi}_{\f{TV}}
.
\]
\end{assumption}

This is a condition on how quickly each worker's chain converges to the target posterior with respect to the behavior of the other workers in the cluster.
It says that, regardless of what the other workers are doing, no worker can proceed at an arbitrarily slow rate of convergence.
We use the term \emph{simultaneous} to emphasize that uniformity is only required with respect to workers, rather than other quantities such as initial conditions of the chain, as is typical in uniform ergodicity and related conditions.
It is through this assumption that properties of the communication scheme, such as how frequently workers transmit their messages to one another, enter the theory.
Whether or not the assumption will hold for a given Gibbs sampler will depend on properties of the target distribution.
Note that at this stage, communication is still instantaneous, and asynchronicity properties such as message delays do not yet enter the theory.
These will be considered later.

\begin{proposition}[Coupled convergence]
\label{prop:sync-gibbs-conv}
Let $\Pi = \bigtimes_{i=1}^m \pi \in E$.
For any instantaneous parallel chain, we have that $H(\Pi) = \Pi$.
Furthermore for all $\eps \in E$ and all $i \in \{1,..,m\}$, the function $||H^k_i(\eps) - \pi||_{\f{TV}}$ is non-increasing in $k$, and we have
\[
d(H^k(\eps), \Pi) \to 0
\]
as $k\to\infty$.
\end{proposition}
\begin{proof}
By additivity of $d$, it suffices to show that all three claims hold for each $H_i$, so fix an arbitrary $i \in \{1,..,m\}$. We have for all $B \in \s{X}$ that
\[
(H_i(\Pi))(B) = \int_\Omega..\int_\Omega H_i(B \given x_1,..,x_m) \d\pi(x_1) .. \d\pi(x_m)
.
\]
Since $H_i$ is non-negative, we may use Tonelli's Theorem to switch the order of integration so that $x_i$ is the inner-most component being integrated.
We then have
\[
\int_\Omega H_i(B \given x_1,..,x_m) \d\pi(x_i) = \pi(B)
\]
because for all $x_1,..,x_m$ except $x_i$, there exists an $l\in L$ and an underlying chain $\xi^k_{(l)}$ with Markov operator $P_{(l)}$ for which we have $H_i = P_{(l)}$.
This gives the first claim.
Next, we check that $||H^k_i(\eps) - \pi||_{\f{TV}}$ is non-increasing in $k$, as well as convergence.
We have that
\<
d(H(\eps), \Pi) &= \sum_{i=1}^m \norm{ \int_\Omega..\int_\Omega H_i(\cdot \given x_1,..,x_m) \d\eps_1(x_1) .. \d\eps_m(x_m) - \pi }_{\f{TV}}
\\
&= \sum_{i=1}^m \norm{ \int_\Omega..\int_\Omega H_i(\cdot \given x_1,..,x_m) \d\eps_i(x_i) \smash{\underbracket[0.5pt]{\d\eps_1(x_1)..\d\eps_m(x_m)}_{\text{except }\d\eps_i(x_i)}} - \int_\Omega..\int_\Omega \pi \, \smash{\underbracket[0.5pt]{\d\eps_1(x_1)..\d\eps_m(x_m)}_{\text{except }\d\eps_i(x_i)}} \, }_{\f{TV}}
\\
&= \sum_{i=1}^m \norm{ \int_\Omega..\int_\Omega P_{(l)}(\eps_i) - \pi \, \smash{\underbracket[0.5pt]{\d\eps_1(x_1)..\d\eps_m(x_m)}_{\text{except }\d\eps_i(x_i)}} \, }_{\f{TV}}
\\
&\leq \sum_{i=1}^m \int_\Omega..\int_\Omega \norm{ P_{(l)}(\eps_i)  - \pi}_{\f{TV}} \smash{\underbracket[0.5pt]{\d\eps_1(x_1)..\d\eps_m(x_m)}_{\text{except }\d\eps_i(x_i)}}
\\
&< \sum_{i=1}^m \int_\Omega..\int_\Omega \rho\norm{ \eps_i  - \pi}_{\f{TV}} \, \smash{\underbracket[0.5pt]{\d\eps_1(x_1)..\d\eps_m(x_m)}_{\text{except }\d\eps_i(x_i)}}
\\
&= \sum_{i=1}^m \rho\norm{ \eps_i  - \pi}_{\f{TV}} \int_\Omega..\int_\Omega \smash{\underbracket[0.5pt]{\d\eps_1(x_1)..\d\eps_m(x_m)}_{\text{except }\d\eps_i(x_i)}}
\\
&= \rho \, d(\eps, \Pi)
.
\>
Here, the second line follows from Tonelli's Theorem since $H_i$ is non-negative, and since the integral of each $\eps_j$ is equal to 1.
The third line follows by linearity and the definition of $P_{(l)}$ in Assumption \ref{asm:simultaneous}.
The fourth line follows from definition of $||\cdot||_{\f{TV}}$, because the supremum of an integral is less than the integral of the supremum.
The fifth line follows from Assumption \ref{asm:simultaneous}.
The sixth line follows because each $\eps_i$ is a probability measure and thus integrates to one.
The last line follows by definition.
Since $\rho < 1$, convergence follows from the Banach fixed point theorem.
\end{proof}

The set of Markov kernels $\{P_{(l)},l\in L\}$, whose properties underly the above analysis, can be viewed as an \emph{adaptive MCMC} algorithm.
From this perspective, the first part of our argument is similar to Proposition 1 of \textcite{roberts07}, and the second part is similar to their Theorem 5, where our Assumption \ref{asm:simultaneous} is similar to their condition (a).

We now move to the second stage of the proof.
From here, we want to show that $H$ converges asynchronously, i.e., convergence is still valid in the setting in which each worker does not necessarily know the precise current state of all other workers, and instead works with the latest state that it knows about.
We begin by stating the \textcite{baudet78, bertsekas83, frommer00} model of distributed computation, within which we base our analysis.

\begin{definition}[Asynchronous computation] \label{def:async-comp}
Start with the following fixed-point computation problem.
\begin{enumerate}[(P1)]
\item Let $E = \bigtimes_{i=1}^m E_i$ be a product space, where $i$ indexes workers. We take $E_i = \c{M}_1(X)$.  \label{a:P1}
\item Let $H : E \to E$ be a function with components $H_i$.\label{a:P2}
\item Let $\Pi = H(\Pi)$ be a fixed point of $H$.
\label{a:P3}
\end{enumerate}

Now, define the following cluster computation model:
\begin{itemize}
\item[--] Let $k \in \N_0$ be the total number of iterations performed by all workers.
\item[--] Let $s_i(k) \in \N_0$ be the total number of iterations on component $i$ by all workers.
\item[--] Let $I_k$ be an index set containing the components updated at iteration $k$.
\end{itemize}

Next, assume the following basic regularity conditions on the cluster:
\begin{enumerate}[(R1)]
\item No worker's state is based on future values: $s_i(k) \leq k-1$. \label{a:R1}
\item No worker stops permanently: $\lim_{k \to \infty} s_i(k) = \infty$. \label{a:R2}
\item No component stops being updated or communicated by workers: $|\{k \in \N : i \in
I^k\}| = \infty$. \label{a:R3}
\end{enumerate}

Finally, define $\eps^k$ component-wise via the following:
\[ \label{asynchronous-computation-1}
\eps_i^k =
\begin{cases}
H_i \! \del[1]{\eps_1^{s_1(k)},..,\eps_m^{s_m(k)}} \text{ if } i \in I^k
,
\\
\eps_i^{k-1} \text{ otherwise.}
\end{cases}
\]
Then $\eps^k$ is termed an \emph{asynchronous iteration}, and $\{\eps^k : k \in \N_0\}$ is termed an \emph{asynchronous computation}.
\end{definition}

Definition \ref{def:async-comp} is broad enough to encompass most asynchronous computations, and it is at this stage that properties such as message delay enter the theory.
With this computational model in mind, the following general theorem gives a sufficient set of conditions under which the asynchronous iterates $\eps^k$ converge to the correct answer.

\begin{result}[Convergence of asynchronous computations] \label{res:async-conv}
Given a well-defined asynchronous computation as in Definition \ref{def:async-comp}, assume the following conditions hold for all $k \in \N_0$:
\begin{enumerate}[(C1)]
\item There are sets $E^k \subseteq E$ satisfying $E^k = \bigtimes_{i=1}^m E_i^k$ \emph{(box condition)}. \label{a:C1}
\item For $E^k$ in (C\ref{a:C1}), $H(E^k) \subseteq E^{k+1} \subseteq E^k$ \emph{(nested sets condition)}. \label{a:C2}
\item There exists a $\Pi$ such that $\eps \in E^k \implies \eps \to \Pi$ in some metric \emph{(synchronous convergence condition)}.\label{a:C3}
\end{enumerate}

Then $\eps^k \to \Pi$ in the same metric.
\end{result}

\begin{proof}
\textcite{baudet78, bertsekas83, frommer00}.
\end{proof}

For MCMC, the main challenge in using this result is that an arbitrary measure space is not a product space -- to avoid this, we instead work with Definition \ref{def:m-coupled-operator}.
We now proceed to verify its conditions.

\begin{lemma}[Box condition] \label{lemma:box}
Fix the initial distribution $\eps \in E$.
Define the following:
\[ \label{asynchronous-computation-3}
E^k = \cbr{ \varpi \in E : ||\varpi_i - \pi||_{\f{TV}} \leq ||H^k_i (\eps) - \pi||_{\f{TV}} \text{ for all } i \in \{1,..,m\} }
.
\]
Then there exist sets $E_i^k$ such that $E^k = \bigtimes_{i=1}^m E_i^k$.
\end{lemma}

\begin{proof}
Take $E_i^k = \{ \mu \in \c{M}_1(X) : ||\mu - \pi||_{\f{TV}} \leq ||H_i^k(\eps) - \pi||_{\f{TV}}\}$.
\end{proof}

\begin{lemma}[Nested sets condition] \label{lemma:nested}
Let $E^k$ be defined as in the previous lemma.
Then $H(E^k) \subseteq E^{k+1} \subseteq E^k$.
\end{lemma}

\begin{proof}
By Proposition \ref{prop:sync-gibbs-conv}, $||H^k_i (\eps) - \pi||_{\f{TV}}$ is non-increasing in $k$ for each $i$, so $E^{k+1} \subseteq E^k$, and
$E^{k+1} = H(E^k)$ by construction.
\end{proof}

\begin{theorem}[Asynchronous convergence]
Asynchronous Markov chains in the sense of Definition \ref{def:sync_chain} and Definition \ref{def:async-comp} satisfying Assumption \ref{asm:simultaneous} converge to $\pi$ on each worker in total variation.
\end{theorem}

\begin{proof}
We verify that all of the conditions required in Result \ref{res:async-conv} hold.
\begin{enumerate}[-----------]
\item[(P1--P3)] Take $E, H, \Pi$ as in Definition \ref{def:sync_chain}.
\item[(R1--R3)] All satisfied by assumption.
\item[(C1--C3)] Satisfied by Lemma \ref{lemma:box}, Lemma \ref{lemma:nested}, and Proposition \ref{prop:sync-gibbs-conv}.
\end{enumerate}
The claim follows.
\end{proof}

\section{Appendix: details of Gibbs sampler and approximate analytic matrix inversion in the Gaussian process example}
\label{apdx:matrixinv}

We propose the following scheme to sample from the posterior of $(\mu, \sigma^2, \tau^2, \v\theta)$.
We update individual slices of $\v\theta$, consisting of 500 elements, via Gibbs steps.
To do this, we sample from full conditional distributions of the form $\v\theta_{1:500 } \given \v\theta_{501:n}, \mu, \sigma^2,$ $\tau^2$ for arbitrary blocks of 500 adjacent indices -- recall that $\phi$ is fixed.
Thus we need to sample from conditional Gaussian distributions of portions of $\v\theta$, given the rest of $\v\theta$.
To do this without ever constructing the large covariance matrix, which may be too big to store in memory, we need to be able to invert $\m{K}$, add $\sigma^{ -2 } \, \m{I}_n$, and invert back.
The following scheme allows us to do this element-wise, with only one approximate inversion along the way, which can with further work likely be refined into an exact inversion.

Since we have made the simplifying assumption that our grid is evenly spaced, the normalized covariance matrix $\tau^{-2}\m{K}$ is Toeplitz.
Additionally, since our covariance function is exponential, the resulting covariance matrix is hyperbolic, and can be inverted element-wise analytically via a technique due to \textcite{dow03}, with inverse that simplifies to

\< \label{eqn:inverse-1}
\tau^2\m{K}^{-1} &=
\begin{bmatrix}
d_0 & a & 0 & \dots & \dots & \dots & 0 \\
a & b & a & 0 & \ddots & \ddots & \vdots \\
0 & a & b & \ddots & \ddots & \ddots & \vdots \\
\vdots & 0 & \ddots & \ddots & \ddots & 0 & \ddots \\
\vdots & \ddots & \ddots & \ddots & b & a & 0 \\
\vdots & \ddots & \ddots & 0 & a & b & a \\
0 & \hdots & \dots & \dots & 0 & a & d_0
\end{bmatrix}
&
& \begin{cases}
b = -\f{coth}(-\phi \rho)
\\
a = \displaystyle\frac{\f{csch}(-\phi \rho)}{2}
\\[1.5ex]
d_0 = \displaystyle\frac{e^{-\phi \rho (2N-3)} \f{csch}(-\phi \rho) + 1 - \f{coth}(-\phi \rho)}{2 - 2 e^{-\phi \rho (2N-3)}}
\\
\rho = \text{grid spacing size} = 0.06
\\
N = \text{dimension of }\m{K}
.
\end{cases}
\>

Note that this $\m{K}^{-1}$ is tridiagonal with modified corner elements.
While this technique limits the generality of our Gaussian process prior, more complicated ways of avoiding large matrix inversions are available with modern spatial priors such as nearest neighbor Gaussian processes \cite{banerjee12}.
If we had not fixed $\phi$, we would have needed to compute a large matrix expression involving $\m{K}^{-1}$ in its entirety for every sample of $\tau^2$ and $\mu$.
Here, this is tractable, but we opted to avoid it for simplicity.

After we add $\sigma^{-2}$ to the diagonal, the resulting covariance matrix is still tridiagonal with modified corner elements.
We do not know how to invert this matrix analytically, but we do know how to invert the general tridiagonal Toeplitz matrix without modified corner elements, via a technique due to \textcite{hu96}.
We approximate the tridiagonal form by assuming that $d_0 = b$ in Equation \eqref{eqn:inverse-1} -- this works well except at the points where the partition slices of $\v\theta$ join, where a small amount of error is introduced.

Finally, to find the mean vector, we need to multiply the covariance matrix defined by Equation \eqref{eqn:inverse-1} by a term that includes the full data.
This multiplication can be carried out to arbitrary precision by simply taking a slice in the center of the matrix in a neighborhood around the full conditional of interest, avoiding use of the full data.
This idea also underlies \emph{covariance tapering} \cite{furrer06} and \emph{composite likelihood methods} for spatial problems \cite{stein04}.
After all of these steps, we can sample any slice of $\v\theta$ full conditionally via the standard Schur complement formula, since the full conditional of a Gaussian is Gaussian.

\section{Appendix: trace plots for hierarchical mixed-effects model of Section \ref{sec:ex3}} \label{apdx:vvb-plots}

\begin{figure*}[h]
\includegraphics[width=0.5\textwidth]{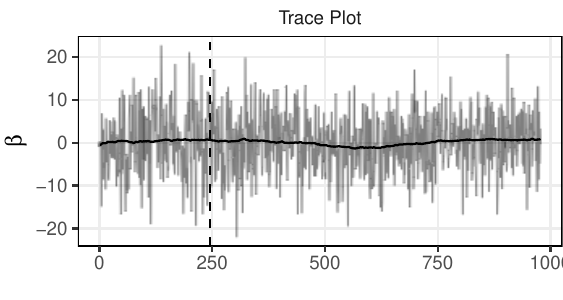}
\includegraphics[width=0.5\textwidth]{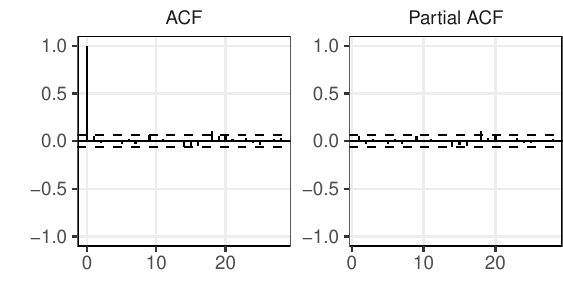}
\includegraphics[width=0.5\textwidth]{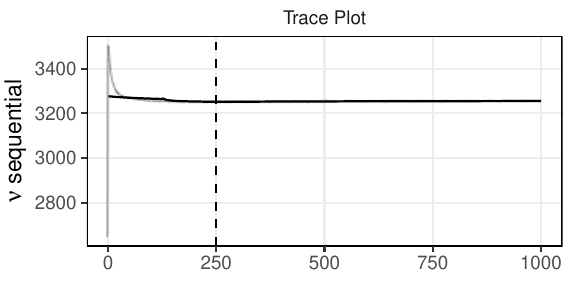}
\includegraphics[width=0.5\textwidth]{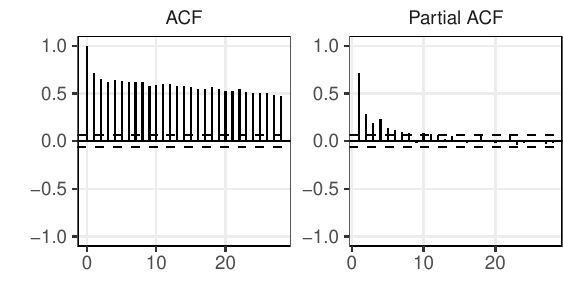}
\includegraphics[width=0.5\textwidth]{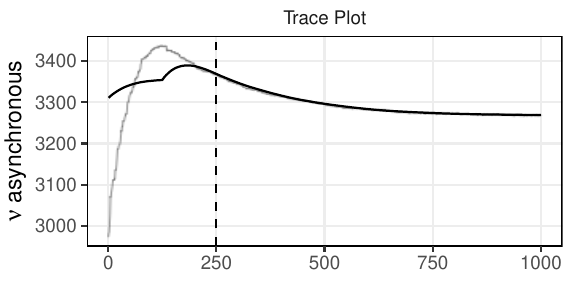}
\includegraphics[width=0.5\textwidth]{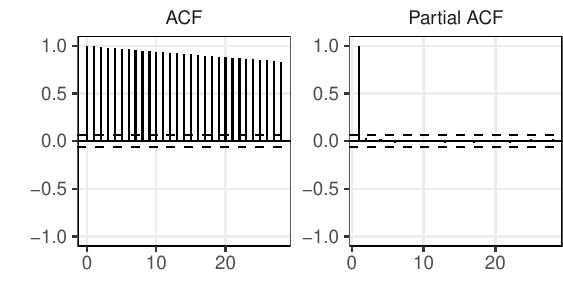}
\caption{Trace plots and autocorrelation plots for an unspecified $\v\beta_i$ component for the asynchronous Gibbs sampler, and of $\nu$ for the asynchronous and sequential-scan Gibbs samplers in Section \ref{sec:ex3}.}
\label{fig:vvb-trace}
\end{figure*}

\end{document}